\documentclass[preprint, 12pt]{article}

\usepackage{amssymb}
\usepackage{amsthm}
\usepackage{amsmath}
\usepackage{graphicx}
\usepackage{hyperref}

\begin{document}




\title{Robustness of the  nonlinear PI control method to ignored actuator dynamics}

\author{Haris E. Psillakis\\
Hellenic Electricity Network Operator S.A.\\
\texttt{psilakish@hotmail.com}}
\date{}
\maketitle
\begin{abstract}
This note examines the robustness properties of the nonlinear PI control method to ignored actuator dynamics. It is proven that global boundedness and regulation can be achieved for sector bounded nonlinear systems with unknown control directions if the actuator dynamics are sufficiently fast and  the nonlinear PI control gain is chosen from a subclass of the Nussbaum function class. Simulation examples are also presented that demonstrate the validity of our arguments.
\end{abstract}






\section{Introduction}
\label{intro}
\newtheorem{theorem}{Theorem}
\newtheorem{lemma}{Lemma}
\newtheorem{remark}{Remark}
\newtheorem{definition}{Definition}
\newtheorem{assumption}{Assumption}
\newtheorem{corollary}{Corollary}
The control problem for systems with unknown control directions has received significant research interest over the last decades \cite{Nussbaum_paper}-\cite{Scheinker2013}. The main solution approach employs Nussbaum functions (NFs) \cite{Nussbaum_paper}-\cite{Yu2013} as control gains with suitable parameter adaptation laws. NFs are continuous functions $N:\mathbb{R}\rightarrow\mathbb{R}$ having the following properties
\begin{align}
    \limsup_{\zeta\rightarrow \pm\infty}&\frac{1}{\zeta}\int_{0}^{\zeta}{N(s)ds}=+\infty\label{nussbaum propertyp}\\
    \liminf_{\zeta\rightarrow \pm\infty}&\frac{1}{\zeta}\int_{0}^{\zeta}{N(s)ds}=-\infty.\label{nussbaum propertym}
\end{align}
Typical examples of NFs are $\zeta^2\sin(\zeta)$ and $\exp(\zeta^4)\cos(\zeta)$ among others.

In \cite{Ortega_paper}, a nonlinear PI control scheme was proposed by Ortega, Astolfi and Barabanov   that also addresses the unknown control direction problem. Its main difference with the Nussbaum methodology is the inclusion of a proportional term in the control gain variable (see \cite{Ortega_paper}, \cite{AKO_book}). Moreover, in the nonlinear PI approach, the Nussbaum property \eqref{nussbaum propertyp}, \eqref{nussbaum propertym} is not a necessary condition and therefore  gains of the form $z\cos(z)$ that do not satisfy \eqref{nussbaum propertyp}, \eqref{nussbaum propertym} can also be used \cite{Ortega_paper}, \cite{AKO_book} with $z$  a PI term of the square error.

Up to now, few results are known for the robustness properties of those schemes  with respect to unmodelled dynamics. In an early paper, Georgiou and Smith \cite{GS} have pointed that the Nussbaum control scheme is nonrobust to fast parasitic first order dynamics (Fig. \ref{GeSm})
 \begin{figure}[!ht]\label{GeSm}
\centering
\includegraphics[width=2.7in]{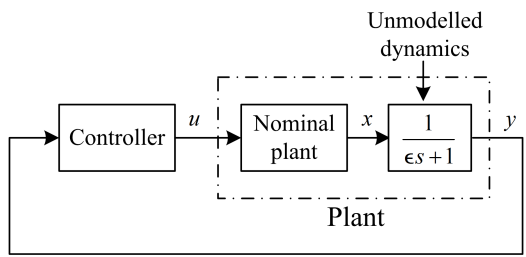}
\caption{The perturbed closed-loop system.}
\end{figure}
for a simple integrator nominal plant. The nonlinear PI controller on the other hand ensures boundedness and regulation in this particular case (simple integrator) as shown in \cite{Ortega_paper}, \cite{AKO_book}. For a nominally unstable plant model with sector bounded nonlinearity, we  proved in \cite{psillakis_SCL} that the nonlinear PI can provide global boundedness and attractivity  only if the PI control gain is a function of Nussbaum type.
In the special case of a perturbed linear system (Fig. \ref{GeSm})
\begin{align}\label{G-SL}
    \dot{x}&=\alpha x+bu\nonumber\\
    \epsilon & \dot{y}=x-y
\end{align}
($b\neq 0, \epsilon>0$) we showed in \cite{psillakis_SCL} that the nonlinear PI controller
 \begin{align}
    u&=\kappa(z)y\label{PI_o}\\
    z&=(1/2)y^2+\lambda\int_0^t{y^2(s)ds}\label{z}
\end{align}
regulates the output to zero if  $\max\{\epsilon\lambda,\epsilon(\alpha+\lambda)\}<1$  and $\kappa(\cdot)$ is a NF (Remark 1 of \cite{psillakis_SCL}). Thus, a combination of the two approaches, i.e. a nonlinear PI controller with a control gain satisfying the Nussbaum property yields improved robustness properties.

\begin{figure}[!ht]
\label{diagram1}
\centering
\includegraphics[width=2.7in]{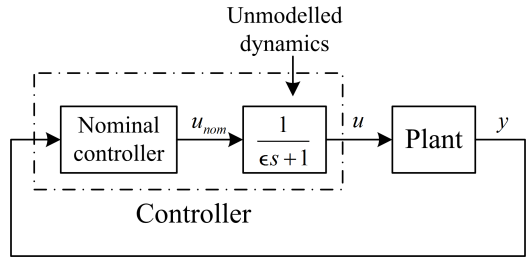}
\caption{The dual perturbed closed-loop system.}
\end{figure}
In this note, we consider \emph{the dual case, namely, the robustness of the  controller to ignored actuator dynamics} (see Fig. 2). To the best of the author's knowledge, this problem has not been treated before in the control literature of systems with unknown control directions.
This is an important issue  since, in many practical cases the fast actuator dynamics are often ignored during the control design. In robot manipulator control for example, the electrical motor dynamics are typically ignored and the joint torques are considered as nominal inputs.

To this end, we examine the  dynamic behavior of  the  nonlinear system  with \emph{first-order unmodelled actuator dynamics} given by
\begin{align}\label{G-SE}
    \dot{y}&=f(y)+bu\nonumber\\
    \epsilon\dot{u}&=u_{nom}-u
\end{align}
for a sector-bounded nonlinear mapping $f$ and  a nonlinear PI control law $u_{nom}$ designed for the nominal system
\begin{equation}\label{unperturbed}
    \dot{y}=f(y)+bu_{nom}.
\end{equation}
In Section \ref{PIpert}, we prove the main contribution of the paper which states that the closed-loop system defined by \eqref{G-SE} and the nonlinear PI controller designed for the nominal system is  \emph{globally bounded with an attractive equilibrium} if the actuator dynamics are sufficiently fast and $\kappa(z)$ belongs to a subclass of the class of Nussbaum functions.

These findings are in complete coherence with the perturbed system model case \cite{psillakis_SCL} showing that,  in both cases, the combination of a nonlinear PI with a Nussbaum type control gain is more robust than the simple nonlinear PI or the Nussbaum gain approach alone.
\subsection{Nonlinear PI control: nominal case}
For system \eqref{unperturbed}, we assume that $f(\cdot)$ is a sector-bounded nonlinearity, i.e.
\begin{align}
    f(y)&=\alpha(y)y\label{sector}\\
\alpha_1\leq \alpha&(y)\leq \alpha_2 \qquad \forall y\in\mathbb{R}\label{sector bounds}
\end{align}
for some constants $\alpha_1, \alpha_2\in\mathbb{R}$.
We further assume $b\neq 0$ for the system to be controllable.
\begin{lemma}
Let the nonlinear system \eqref{unperturbed} with nonlinearity \eqref{sector}, \eqref{sector bounds}. Consider also the nonlinear PI controller of the form
\begin{align}
    u_{nom}&=\kappa(z)y\label{unom}
\end{align}
\begin{align}
    z=\frac{1}{2}y^2+\lambda\int_0^t{y^2(s)ds}\label{zPInom}
\end{align}
($\lambda>0$) with PI gain $\kappa(z):=\beta(z)\cos(z)$ and $\beta(\cdot)$ a class $\mathcal{K}_{\infty}$ function  \footnote{A function $\beta(\cdot)$ belongs to class $\mathcal{K}_{\infty}$ if it is continuous, strictly increasing with $\beta(0)=0$ and $\lim_{z\rightarrow+\infty}\beta(z)=+\infty$.}. Then, for the closed-loop system we have that $z,y,u_{nom}$ are bounded and $\lim_{t\rightarrow\infty}y(t)=\lim_{t\rightarrow\infty}u_{nom}(t)=0$.
\end{lemma}
\begin{proof}
The proof is a simple generalization of the results of section 1.1 in \cite{psillakis_SCL} and is therefore omitted.
\end{proof}

\section{Nonlinear PI control: ignored actuator dynamics case}
\label{PIpert}
Assume now the existence of parasitic first order unmodelled actuator dynamics in the form of \eqref{G-SE} with  sector-bounded nonlinearity \eqref{sector}, \eqref{sector bounds}. The main result of the paper is given below.
\begin{theorem}\label{main_theorem}
Let the closed-loop system described by \eqref{G-SE},  \eqref{unom}, \eqref{zPInom} with sector-bounded nonlinearity given by \eqref{sector}, \eqref{sector bounds}. If
\begin{description}
  \item[(i)]  $\epsilon(\lambda+\alpha_2)<1$
  \item[(ii)] $\kappa(z)=\beta(z)\cos(z)$ with $\beta(\cdot)$ a $\mathcal{K}_{\infty}$ function having the  property 
\begin{equation}\label{beta_property}
    \lim_{z\rightarrow +\infty}\bigg[\frac{\beta(z+\epsilon)}{z}-c\beta(z)\bigg]=+\infty\quad,\:\forall c,\epsilon>0
\end{equation}
\end{description}
then, all closed-loop signals are bounded and $\lim_{t\rightarrow\infty}y(t)=\lim_{t\rightarrow\infty}u(t)=\lim_{t\rightarrow\infty}u_{nom}(t)=0.$
\end{theorem}
\begin{remark}
Property \eqref{beta_property} is satisfied for example for the function $\beta(z)=c_1[\exp(c_2 z^2)-1]$ for all $c_1,c_2>0$. Note that the function $\kappa(z)=\beta(z)\cos(z)$ (with $\beta(\cdot)$ some $\mathcal{K}_{\infty}$ function having the  property \eqref{beta_property}) is a Nussbaum function satisfying \eqref{nussbaum propertyp},\eqref{nussbaum propertym}, i.e. the function $\kappa(z)$ described by (ii) belongs to a special subclass of the class of all Nussbaum functions.
\end{remark}
\begin{proof}
From the definition of the PI error $z$ in \eqref{zPInom} and \eqref{G-SE} we have that
\begin{equation}\label{z_dynamics}
    \dot{z}=byu+(\alpha(y)+\lambda)y^2.
\end{equation}
Let now the function
\begin{align}\label{S}
    S(u,y):=\frac{\epsilon}{2}u^2+\frac{\epsilon(\alpha_2+\lambda)}{b}uy+\frac{\ell}{2}y^2.
\end{align}
with $\ell$ some positive constant to be defined. Replacing from \eqref{G-SE}, \eqref{unom}, \eqref{zPInom}, \eqref{z_dynamics} and canceling terms we have for its time derivative that
\begin{align}\label{dotS}
    \dot{S}=-[1-\epsilon(\alpha_2+\lambda)]u^2&-\frac{1}{b}(\alpha_2+\lambda)(1-\epsilon\alpha(y))uy-\lambda\ell y^2\nonumber\\
    &+\frac{1}{b}(\alpha_2-\alpha(y))\kappa(z)y^2+\ell \dot{z}+\frac{1}{b}\kappa(z)\dot{z}.
\end{align}
Eq. \eqref{dotS} can be written in matrix notation as
\begin{align}\label{dotS1}
    \frac{d}{dt}\bigg[S-\frac{1}{b}\int_0^z{(\kappa(s)+b\ell)ds}\bigg]=-w^T\Lambda(y)w+\frac{1}{b}(\alpha_2-\alpha(y))\kappa(z)y^2
\end{align}
with $w=\left[
          \begin{array}{cc}
            u & y \\
          \end{array}
        \right]^T$ and
\begin{equation}\label{Lambda}
    \Lambda(y):=\left[
    \begin{array}{cc}
    1-\epsilon(\lambda+\alpha_2) & \frac{1}{2b}(\lambda+\alpha_2)(1-\epsilon\alpha(y))\\
    * & \lambda\ell
    \end{array}\right]
\end{equation}
where $*$  denotes a symmetric w.r.t. the main diagonal element of $\Lambda(y)$. From 
the definition of $S(u,y)$ and $\Lambda(y)$ it is obvious that if we select  a sufficiently large constant $\ell$
\begin{equation}\label{ell}
    \ell>\bigg(\frac{\alpha_2+\lambda}{b}\bigg)^2\max\bigg\{\epsilon,\frac{(1-\epsilon\alpha_1)^2}{4\lambda\big[1-\epsilon(\lambda+\alpha_2)\big]}\bigg\}.
\end{equation}
and $\epsilon(\alpha_2+\lambda)<1$ then $S(u,y)\geq 0$ for all $(u, y)\in\mathbb{R}^2$ and
$\Lambda(y)$ is positive definite $\forall y\in\mathbb{R}$.  Then, from \eqref{dotS1} we have
\begin{align}\label{Sb1}
    S(t)+\int_0^t&{w^T(s)\Lambda(y(s))w(s)ds}\leq S(0)+\ell z(t)\nonumber\\
    &+\frac{1}{b}\int_0^{z(t)}{\kappa(s)ds}+\frac{1}{b}\int_0^t{(\alpha_2-\alpha(y(s)))\kappa(z(s))y^2(s)ds}.
\end{align}

We claim now that $z$ is bounded. Assume the opposite, i.e. that $z$ grows unbounded. Then, as $z$ progresses to infinity, consider the sequences of times 
$\{t_{1k}\}$,$\{t_{2k}\}$ defined by
\begin{align}\label{tik}
    t_{2k}&:=\inf\{t\in\mathbb{R}:z(t)=z_{2k}\}\\
    t_{1k}&:=\sup\{t\in[0,t_{2k}):z(t)=z_{1k}\}
\end{align}
with
\begin{align}\label{zik}
    z_{1k}:=2\pi k+(\pi/2)(1+\textrm{sgn}(b))-\pi/2\\
    z_{2k}:=2\pi k+(\pi/2)(1+\textrm{sgn}(b))+\pi/4
\end{align}
\begin{figure}[!ht]
\label{zt_red}
\centering
\includegraphics[width=2.5in]{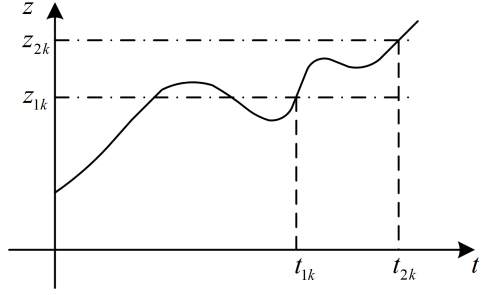}
\caption{Times  $t_{1k}$, $t_{2k}$.}
\end{figure}
(see Fig. 3). From the definitions above, we have that $z(t)\in[z_{1k},z_{2k}]$ for all $t\in[t_{1k},t_{2k}]$. Since we have assumed that $z$ is unbounded, $z$ will eventually pass sequentially from an infinite number of consecutive elements of the sequences $\{z_{1k}\}_{k=k_0}^{\infty}$, $\{z_{2k}\}_{k=k_0}^{\infty}$ at times 
$\{t_{1k}\}_{k=k_0}^{\infty}$ and $\{t_{2k}\}_{k=k_0}^{\infty}$ respectively with $k_0$ some positive integer determined from the  initial conditions. Note that whenever $z\in[z_{1k},z_{2k}]$ then $(1/b)\kappa(z)\leq 0$ and therefore $(1/b)\kappa(z(t))\leq 0$ for all $t\in[0,t_{2k}]$ for which $z(t)\geq z_{1k}$.
Hence, for the last term in the r.h.s. of \eqref{Sb1} we have the following upper bound for $t=t_{1k}$
\begin{align}\label{intBart1kt1k}
    \int_{0}^{t_{1k}}{\frac{\alpha_2-\alpha(y(s))}{b}\kappa(z(s))y^2(s)ds}&\leq \int_{\substack{t\in[0,t_{1k}]\\ z(t)\leq z_{1k}}}{\frac{\alpha_2-\alpha(y(s))}{b}\kappa(z(s))y^2(s)ds}\nonumber\\
    &\leq \frac{\alpha_2-\alpha_1}{|b|} \sup_{\substack{t\in[0,t_{1k}]\\ z(t)\leq z_{1k}}}\big\{\beta(z(t))\big\}\int_{\substack{t\in[0,t_{1k}]\\ z(t)\leq z_{1k}}}{y^2(s)ds}\nonumber\\
    &\leq \frac{\alpha_2-\alpha_1}{\lambda|b|}\beta(z_{1k})z_{1k}
\end{align}
and for $t=t_{2k}$
\begin{align}\label{int0t2k}
    \int_{0}^{t_{2k}}{\frac{\alpha_2-\alpha(y(s))}{b}\kappa(z(s))y^2(s)ds}&\leq \int_{0}^{t_{1k}}{\frac{\alpha_2-\alpha(y(s))}{b}\kappa(z(s))y^2(s)ds}\nonumber \\
    &\leq \frac{\alpha_2-\alpha_1}{\lambda|b|}\beta(z_{1k})z_{1k}.
\end{align}
as $\textrm{sgn}(b)\kappa(z(t))\leq 0$ for all $t\in [t_{1k},t_{2k}]$. Choosing also $t=t_{2k}$ in \eqref{Sb1} and taking into account \eqref{int0t2k} we arrive at
\begin{align}\label{Sb}
    S(t_{2k})\leq S(0)+\frac{1}{b}\int_0^{z_{2k}}{(\kappa(s)+b\ell)ds}+\frac{\alpha_2-\alpha_1}{\lambda|b|}\beta(z_{1k})z_{1k}.
\end{align}
Considering that $\textrm{sgn}(b)\kappa(z)\leq 0$ for $z\in[z_{1k},z_{2k}-\pi/2]$ and $\textrm{sgn}(b)\kappa(z)\leq -(1/\sqrt{2})\beta(z_{2k}-\pi/2)$ for $z\in[z_{2k}-\pi/2,z_{2k}]$ we have
\begin{align}\label{Nb}
    \frac{1}{b}\int_0^{z_{2k}}{(\kappa(s)+b\ell)ds}&\leq \ell z_{2k}+\frac{1}{b}\int_0^{z_{1k}}{\kappa(s)ds}+\frac{1}{b}\int_{z_{2k}-\pi/2}^{z_{2k}}{\kappa(s)ds}\nonumber\\
    &\leq \ell z_{2k}+\frac{1}{|b|}\beta(z_{1k})z_{1k}-\frac{\pi}{2\sqrt{2}|b|}\beta(z_{2k}-\pi/2).
\end{align}
Taking into account the fact that $z_{2k}=z_{1k}+3\pi/4$ and combining \eqref{Sb} and \eqref{Nb} we obtain
\begin{align}\label{Sbound}
   S(t_{2k})\leq &S(0)+\frac{3\ell\pi}{4}+\ell z_{1k}\nonumber\\
   &+\frac{1}{|b|}\bigg(1+\frac{\alpha_2-\alpha_1}{\lambda}\bigg)\beta(z_{1k})z_{1k}-\frac{\pi}{2\sqrt{2}|b|}\beta\bigg(z_{1k}+\frac{\pi}{4}\bigg).
\end{align}
From property \eqref{beta_property} of $\beta(\cdot)$, as  $k$ tends to infinity, the r.h.s. of \eqref{Sbound} becomes negative for sufficiently large values of $k$ for arbitrary values of the associated constants $b$, $\lambda$, $\ell$, $\alpha_1$, $\alpha_2$ and $S(0)$. This, however, will enforce negative values to $S(t_{2k})$ in the l.h.s. of \eqref{Sbound} that is not possible since $S$ is a nonnegative function. Thus, there is a contradiction and $z$ is bounded, i.e. $z\in\mathcal{L}_{\infty}$. Equivalently we have that $y\in\mathcal{L}_{\infty}\cap\mathcal{L}_{2}$. Then, from \eqref{Sb1} we obtain $S\in \mathcal{L}_{\infty}$, $u\in\mathcal{L}_{\infty}\cap\mathcal{L}_{2}$.
Furthermore, the system equations \eqref{G-SE} yield $\dot{u},\dot{y}\in\mathcal{L}_{\infty}$. Invoking now Barbalat lemma we result in $\lim_{t\rightarrow\infty}y(t)=\lim_{t\rightarrow\infty}u(t)=0$ that also yields $\lim_{t\rightarrow\infty}u_{nom}(t)=0$.
\end{proof}
For the linear system case, condition (ii) of Theorem \ref{main_theorem} can be reduced to $\kappa(\cdot)$ satisfying \eqref{nussbaum propertyp}, \eqref{nussbaum propertym} as stated in the following corollary.
\begin{corollary}
Let the closed-loop system described by the linear system with ignored fast actuator dynamics
\begin{align}\label{linear_perturbed}
    \dot{y}&=\alpha y+bu\nonumber\\
    \epsilon\dot{u}&=u_{nom}-u
\end{align}
and controller \eqref{unom}, \eqref{zPInom}. If   $\epsilon(\lambda+\alpha)<1$
and   $\kappa(\cdot)$ is a NF then, all closed-loop signals are bounded and $\lim_{t\rightarrow\infty}y(t)=\lim_{t\rightarrow\infty}u(t)=\lim_{t\rightarrow\infty}u_{nom}(t)=0$.
 \end{corollary}
 \begin{proof}
 In the case of a linear system $\alpha(y)=\alpha_1=\alpha_2=\alpha$ and the last integral in the r.h.s. of \eqref{Sb1} is equal to zero.  From the NF assumption and \eqref{Sb1} boundedness of $z$ can then be proved. The rest of the proof continues along the lines of the proof of Theorem \ref{main_theorem} and is therefore omitted.
 \end{proof}
\begin{remark}
Simulation tests (see Section \ref{simulation}) for the nonlinear PI controller with a gain that is not a NF and the standard NF-based controller reveal that in both cases divergent output trajectories can occur even if $\epsilon(\alpha+\lambda)<1$. Thus, the proposed nonlinear PI control law with NF gain is more robust than the alternative approaches for both the case of ignored actuator dynamics and the case of  unmodelled system dynamics treated in \cite{psillakis_SCL}.
\end{remark}
 \section{Simulation examples}
 \label{simulation}
\subsection{Linear system}
A simulation study was performed for the linear system with ignored actuator dynamics (LSIAD) described by  \eqref{linear_perturbed} with parameters $\alpha=0.8$, $b=0.05$, $\epsilon=0.1$ and initial conditions $y(0)=5$, $u(0)=0$. We tested the case of a Nussbaum gain based (NG) controller
\begin{align}\label{nussbaum_controller}
    u_{nom}&=\zeta^2\cos(\zeta)y\nonumber\\
    \dot{\zeta}&=\lambda y^2
\end{align}
with $\zeta(0)=0$, $\lambda=0.15$ and a nonlinear PI controller \eqref{unom}, \eqref{zPInom} with gains $\kappa(z)=z\cos(z)$  (not a Nussbaum function) denoted as nPI and $\kappa(z)=z^2\cos(z)$ (Nussbaum function) denoted as nPI-N. For the specific selection of parameter $\lambda$, condition (i) of Theorem \ref{main_theorem} holds true.
\begin{figure}[!ht]
\centering
\includegraphics[width=3.7in]{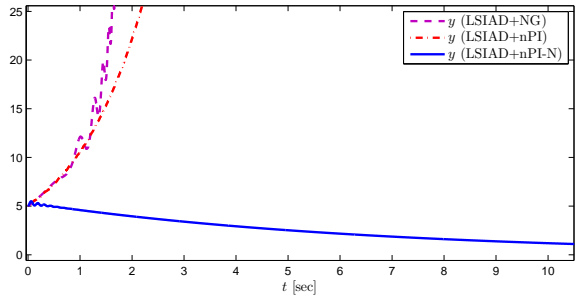}
\caption{Output responses $y(t)$ of a  linear system with ignored actuator dynamics (LSIAD) and the  controllers NG, nPI, nPI-N.}
\label{y_linear}
\end{figure}
The output trajectories $y(t)$ for the three cases are shown in Fig. \ref{y_linear}. We observe divergent output responses  for both the NG and the nPI controllers. Convergent solutions are only obtained when the nPI-N cpntroller is employed.

\subsection{Nonlinear system}
Let now the perturbed nonlinear system \eqref{G-SE} with $f(x)=3[1+2\sin(\exp(x))]x$ and $b=1$. From the definition of $f$ the sector bounds are $\alpha_1=-3$ and $\alpha_2=9$. Selecting now $\lambda=0.5$, condition (i) of Theorem 1 is satisfied for every $\epsilon<1/(\alpha_2+\lambda)=0.105$.  We apply the  control law \eqref{PI_o}, \eqref{z} with $\kappa(z)=\big[\exp(z^2/10)-1\big]\cos(z)$ satisfying condition (ii) of our main Theorem. Simulation results are shown in Fig. \ref{sector_example} with $\epsilon=0.1$  and initial conditions $u(0)=0$, $y(0)=5$.
\begin{figure}[!ht]
\centering
\includegraphics[width=3.9in]{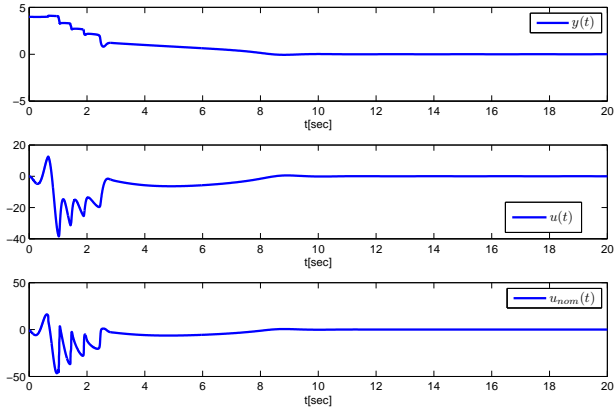}
\caption{Nonlinear system: Time responses of  $y(t),u(t)$ and $u_{nom}(t)$.}
\label{sector_example}
\end{figure}
We observe that the time responses of $y(t),u(t),u_{nom}(t)$ are bounded and converge to the origin  as expected from our theoretical analysis.
 \section{Conclusions}
Robustness to ignored actuator dynamics of the nonlinear PI control for a class of sector bounded nonlinear systems with unknown control direction is analyzed in this note. Selecting the PI control gain from a subclass of the NF class, we prove global boundedness and regulation  for sufficiently fast ignored actuator dynamics.
\label{}

\end{document}